\documentclass[twocolumn]{article}

\usepackage{algorithm}
\usepackage{algorithmic}
\usepackage{amsmath}
\usepackage{amsthm}
\usepackage{amssymb}
\usepackage{graphicx}
\usepackage{epstopdf}
\usepackage{balance}
\usepackage{times}
\usepackage{xspace}
\usepackage{url}
\usepackage{xcolor}
\usepackage{multirow}
\usepackage{fullpage}

\title{Aggregate Skyline Join Queries: Skylines with Aggregate Operations over
Multiple Relations}

\author{
\begin{tabular}{cc}
	Arnab Bhattacharya & B. Palvali Teja \\
	\url{arnabb@iitk.ac.in} & \url{palvali.teja@gmail.com} \\
	{Dept. of Computer Science and Engineering,} & {Amazon Development Limited,} \\
	{Indian Institute of Technology, Kanpur,} & {Hyderabad,} \\
	{Kanpur, India} & {India} \\
\end{tabular}
}

\date{}

\newtheorem{definition}{Definition}
\newtheorem{lemma}{Lemma}
\newtheorem{theorem}{Theorem}

\newcommand{\figwidth}{0.65\columnwidth}
\newcommand{\comment}[1]{}

\begin{document}

\maketitle

\begin{abstract}
	The multi-criteria decision making, which is possible with the advent of
	skyline queries, has been applied in many areas.  Though most of the
	existing research is concerned with only a single relation, several real
	world applications require finding the skyline set of records over multiple
	relations.  Consequently, the join operation over skylines where the
	preferences are local to each relation, has been proposed.  In many of those
	cases, however, the join often involves performing aggregate operations
	among some of the attributes from the different relations.  In this paper,
	we introduce such queries as ``aggregate skyline join queries''.  Since the
	na\"ive algorithm is impractical, we propose three algorithms to efficiently
	process such queries.  The algorithms utilize certain properties of skyline
	sets, and processes the skylines as much as possible locally before
	computing the join.  Experiments with real and synthetic datasets exhibit
	the practicality and scalability of the algorithms with respect to the
	cardinality and dimensionality of the relations.\\
\end{abstract}


\section{Introduction}
\label{sec:intro}

The skyline operator, introduced by B\"orzs\"onyi et al.~\cite{656550},
addresses the problem of multi-criteria decision making where there is no clear
preference function over the attributes, and the user wants an overall big
picture of which objects dominate (equivalently, are better than) other objects
in terms of preferences set by her.  The classic example involves choosing
hotels that are good in terms of both price and distance to beach.  The
\emph{skyline} set of hotels discard other hotels that are both dearer and
farther than a skyline hotel.

For every attribute, there is a preference function that states which objects
dominate over other objects.  For example, the preference function for both
price and distance to beach is $<$, i.e., a hotel with a lower price \emph{and}
at a closer distance to the beach than another hotel will dominate the second
one.  Consequently, the second hotel is never going to be preferred, and does
not require any further consideration.  The skyline query returns all such
objects that are \emph{not} dominated by any other object.  The importance and
usefulness of skyline queries has provoked the commercial database management
systems to incorporate these queries into existing systems~\cite{1129924}.

In real applications, however, there often exists a scenario when a single
relation is not sufficient for the application, and the skyline needs to be
computed over multiple relations~\cite{4354442}.  For example, consider a flight
database.  A person traveling from city A to city B may use stopovers, but may
still be interested in flights that are cheaper, have a less overall journey
time, better ratings and more amenities.  In this case, a single relation
specifying all direct flights from A to B may not suffice or may not even exist.
The join of multiple relations consisting of flights starting from A and those
ending at B needs to be processed before computing the preferences.

The above problem becomes even more complex if the person is interested in the
travel plan that optimizes both on the \emph{total} cost as well as the
\emph{total} journey time for the two flights (other than the ratings and
amenities of each airline).  In essence, the skyline now needs to be computed on
attributes that have been \emph{aggregated} from multiple relations in addition
to attributes whose preferences are \emph{local} within each relation.  The
common aggregate operations are sum, average, minimum, maximum, etc.

Table~\ref{tab:example} shows an example.  The first table lists all flights
from city A and the second one lists all flights to city B.  A join of the two
tables with the destination of the first table equal to the source of the second
table and departure time more than arrival time will yield all flights from A to
B with one stopover.  As shown in Table~\ref{tab:example}(c), it also contains
the total cost, total journey time, ratings and amenities of the two flights.
The user wants a skyline on this joined relation using these attributes.  While
the total cost and total journey time are aggregated attributes, the ratings and
amenities are local to each table.  In this example, flight (13, 23) is
dominated by flight (11, 21) in all the attributes, and hence, will not be
preferred.  On the other hand, flight (11, 21) is not dominated by any other
flight and, therefore, is part of the skyline set that the user wants to examine
it more thoroughly.  We name the above queries that retrieve skylines over
aggregates of attributes joined using multiple relations as \textsc{Aggregate
Skyline Join Queries (ASJQ)}.

The above query can be specified in SQL as:
\begin{verbatim}
SELECT f1.fno, f2.fno,
  f1.dst, f2.src, f1.arr, f2.dep,
  f1.rtg, f2.rtg, f1.amn, f2.amn,
  cost as f1.cost + f2.cost,
  duration as f1.duration + f2.duration
FROM FlightsA as f1, FlightsB as f2
WHERE f1.dst = f2.src AND
      f1.arr < f2.dep AND
  SKYLINE of cost min, duration min,
    f1.rtg max, f2.rtg max,
    f1.amn max, f2.amn max
\end{verbatim}
%
Thus, database systems that have the \emph{skyline} operator built into
them~\cite{skylinepostgres} can easily allow the users to run such queries.

\begin{table*}
	\begin{center}
		\begin{tabular}{cc}
			\hspace*{-3mm}
			\begin{tabular}{|c|c|c|c|c|c|c|c|}
				\hline
				& & \multicolumn{2}{c|}{Join (H)} & \multicolumn{2}{c|}{Aggregate (G)} &
				\multicolumn{2}{c|}{Local (L)} \\
				\cline{3-8}
				fno & dep & arr & dst & duration & cost & amn & rtg \\
				\hline
				11 & 06:30 & 08:40 & C & 2h 10m & 162 & 5 & 4 \\
				12 & 07:00 & 09:00 & E & 2h 00m & 166 & 4 & 5 \\
				14 & 08:05 & 10:00 & E & 1h 55m & 140 & 3 & 4 \\
				15 & 09:50 & 10:40 & C & 1h 40m & 270 & 3 & 2 \\
				13 & 12:00 & 13:50 & C & 1h 50m & 173 & 4 & 3 \\
				16 & 16:00 & 17:30 & D & 1h 30m & 230 & 3 & 3 \\
				17 & 17:00 & 20:20 & C & 3h 20m & 183 & 4 & 3 \\
				\hline
				\multicolumn{8}{c}{(a) Flights from city A (FlightsA)}
			\end{tabular}
			& \hspace*{-4mm}
			\begin{tabular}{|c|c|c|c|c|c|c|c|}
				\hline
				& \multicolumn{2}{c|}{Join (H)} & & \multicolumn{2}{c|}{Aggregate (G)} &
				\multicolumn{2}{c|}{Local (L)} \\
				\cline{2-3} \cline{5-8}
				fno & src & dep & arr & duration & cost & amn & rtg \\
				\hline
				21 & C & 09:50 & 12:00 & 2h 10m & 162 & 5 & 4 \\
				26 & C & 16:00 & 18:49 & 2h 49m & 160 & 2 & 3 \\
				23 & C & 16:00 & 18:45 & 2h 45m & 160 & 4 & 4 \\
				25 & D & 16:00 & 17:49 & 1h 49m & 220 & 3 & 4 \\
				22 & D & 17:00 & 19:00 & 2h 00m & 166 & 4 & 5 \\
				27 & E & 20:00 & 21:46 & 1h 46m & 200 & 3 & 3 \\
				24 & E & 20:00 & 21:30 & 1h 30m & 160 & 4 & 3 \\
				\hline
				\multicolumn{8}{c}{(b) Flights to city B (FlightsB)}
			\end{tabular}
			\\\\
			\multicolumn{2}{c}{
			\hspace*{-3mm}
			\begin{tabular}{|c|c|c|c|c|c|c|c|c|c|c|c|c|}
				\hline
				f1.fno & f2.fno & f1.dst & f2.src & f1.arr & f2.dep & f1.amn & f2.amn 
				& f1.rtg & f2.rtg & cost & duration & Skyline \\
				\hline
				11 & 21 & C & C & 08:40 & 09:50 & 5 & 5 & 4 & 4 & 324 & 4h 20m & Yes \\
				11 & 23 & C & C & 08:40 & 16:00 & 5 & 4 & 4 & 4 & 322 & 4h 55m & Yes \\
				13 & 23 & C & C & 13:50 & 16:00 & 4 & 4 & 3 & 4 & 333 & 4h 35m & No \\
				15 & 23 & C & C & 10:40 & 16:00 & 3 & 4 & 2 & 4 & 430 & 4h 25m & No \\
				12 & 24 & E & E & 09:00 & 20:00 & 4 & 4 & 5 & 3 & 326 & 3h 30m & Yes \\
				14 & 24 & E & E & 10:00 & 20:00 & 3 & 4 & 4 & 3 & 300 & 3h 25m & Yes \\
				\hline
				\multicolumn{11}{c}{(c) Part of the joined relation (FlightsA $\Join$ FlightsB)}
			\end{tabular}
			}
		\end{tabular}
	\end{center}
	\caption{Example of an Aggregate Skyline Join Query (ASJQ).}
	\label{tab:example}
\end{table*}

The preferences in the general skyline join problems are local to each relation,
and hence, the skyline operations can be performed before the
join~\cite{4354442}.  For ASJQ queries, however, the skyline is computed over
the aggregate values of attributes from multiple relations.  This leads to
performance degradation, since, the cardinality of joined relations is in
general large, and the skylines cannot be processed unless the aggregate values
have been computed.  The aggregation function must be \emph{monotonic}, i.e., if
values $s$ and $u$ are preferred over values $t$ and $v$ respectively, the
aggregated value of $s$ and $u$ must be preferred over the aggregated value of
$t$ and $v$.  The aggregation operation is reminiscent of the problem of finding
top-$k$ objects using multiple sources~\cite{375567}.  However, the ASJQ queries
differ significantly by retrieving the skylines in which the aggregate values
are only part of the attribute set.  ASJQ queries, thus, involve three separate
problems---skyline queries, join and aggregation from multiple
sources---together, and highlights the connections among them.

The ASJQ queries are pertinent in many application domains.  For example, the
situation with flights described above is quite a routine task for tour planners
and traveling salespersons.  Another interesting application is in the cricket
leagues.  Clubs want to buy both good batsmen and good bowlers.  Batsmen have
attributes such as average, cost and rating.  Similarly, bowlers have strike
rate, cost and rating.  The clubs optimize their chances of winning by
considering options from the skyline set of batsman-bowler combinations with
preferences for high average, high strike rate, low total cost and high total
rating.  In the same way, to choose an optimal combination of digital camera and
a compatible memory card from a products database, it is necessary to join the
individual tuples containing the attributes of a camera and those of a memory
card on an attribute such as compatible memory card type (e.g., SD, XD, CF
etc.), and optimize an aggregate attribute such as total cost, in addition to
local attributes such as optical zoom (for camera) and storage capacity (for
memory card).  ASJQ queries can also be applied in the context of multimedia
data retrieval~\cite{375567}, geographic information systems~\cite{263689},
dynamic resource allocation on the grid~\cite{823372}, e-commerce~\cite{375739},
etc.

The na\"ive method of implementing ASJQ involves three steps: (i)~performing the
join operation over the relations, (ii)~performing the aggregate operations on
the attributes of multiple relations, and (iii)~performing the skyline query on
the joined relation.  For large relations, this demands impractical
computational costs.  By intuition, one can observe that non-skyline points in
each relation cannot appear in the final result set.  Hence, performing a
skyline operation on each relation before joining reduces the size of the
relations to be joined and, thus, reduces the processing cost.

To reduce the costs further, we designed three algorithms.  The first approach,
\emph{Multiple Skyline Computations (MSC)} algorithm, utilizes the fact that
certain joins of non-skyline sets from the individual relations need not be
tested for skyline criteria, and can be pruned.  The \emph{Dominator-based}
algorithm and the \emph{Iterative} algorithm improve on the MSC approach by
pruning records even from the skyline sets of individual relations before they
are joined, and are thus more efficient.

Our contributions in this paper are:
\begin{enumerate}
	\item We define a novel query ``Aggregate Skyline Join Query''.
	\item We propose three algorithms that efficiently solves them.
	\item We thoroughly investigate the effects of different parameters on the
		algorithms in terms of computational costs both analytically and through
		experiments.
\end{enumerate}

The rest of the paper is organized as follows.   The Aggregate Skyline Join
Query (ASJQ) is formally defined in Section~\ref{sec:statement}.  A brief
literature review is presented in Section~\ref{sec:work}.
Section~\ref{sec:algo} proposes and analyzes three algorithms that efficiently
solves the ASJQ queries.  Section~\ref{sec:exp} describes the experimental
results before Section~\ref{sec:concl} concludes.

\section{Problem Statement}
\label{sec:statement}

We begin by recapitulating the definition of skyline queries for a relation.
Certain attributes of the relation participate in the skyline and are called the
skyline attributes.  For each skyline attribute, \emph{preference functions} are
specified as part of the skyline query.  In a relation $R$, a tuple $r_i =
(r_{i_1}, r_{i_2}, \dots, r_{i_k})$ \emph{dominates} another tuple $r_j =
(r_{j_1}, r_{j_2}, \dots, r_{j_k})$, denoted by $r_i \succ r_j$, if for all
skyline attributes $c = \{s_1, \dots, s_{k'}\} \subseteq \{1, \dots, k\}$,
$r_{i_c}$ is preferred over or equal to $r_{j_c}$, and there is at least one
attribute $d$ where $r_{i_d}$ is strictly preferred over $r_{j_d}$.  A tuple $r$
is in the \emph{skyline} set of $R$ if there does not exist any tuple $s \in R$
that dominates $r$.

For our problem, i.e., ASJQ, the attributes of a relation are categorized into
three types: (i)~\emph{local (L)}: attributes on which skyline preferences are
applied locally to each relation, (ii)~\emph{aggregate (G)}: attributes on which
skyline preferences are applied \emph{after} the aggregate operations are
performed during join, (iii)~\emph{join (H)}: attributes on which no skyline
preferences are specified, but are instead used for joining the two relations.

\begin{definition}[Local attributes]
	The attributes of a relation on which preferences are applied for the
	purposes of skyline computation, but no aggregate operation with an
	attribute from the other relation is performed, are denoted as \emph{local
	attributes}.
\end{definition}

\begin{definition}[Aggregate attributes]
	The attributes of a relation, on which an aggregate operation is performed
	with another attribute from the other relation, and then preferences are
	applied on the aggregated value for skyline computation, are denoted as
	\emph{aggregate attributes}.
\end{definition}

\begin{definition}[Join attributes]
	The attributes of a relation, on which no skyline preferences are specified,
	but are used to specify the join conditions between the two relations, are
	denoted as \emph{join attributes}.
\end{definition}

Denoting the local attributes by $l$, the aggregate attributes by $g$, and the
join attributes by $h$, the two relations can be represented as:
\begin{align}
	R_1 &= \{ h_{1_1}, \dots, h_{1_j}, l_{1_1}, \dots, l_{1_{m_1}},
	g_{1_1}, \dots, g_{1_n} \} \nonumber \\
	R_2 &= \{ h_{2_1}, \dots, h_{2_j}, l_{2_1}, \dots, l_{2_{m_2}},
	g_{2_1}, \dots, g_{2_n} \} \nonumber
\end{align}
where $R_1$ and $R_2$ has $m_1$ and $m_2$ local attributes respectively, and $n$
aggregate attributes.  The join condition is a conjunction of $j$ comparisons
between the corresponding $j$ attributes ($h_{ij}$) of $A$ and $B$.  In this
paper, we assume that join attributes are separate from local and aggregate
attributes.  The final joined relation $R = R_1 \Join R_2$ is
\begin{align}
	R = \{ & h_{1_1}, \dots, h_{1_j}, h_{2_1}, \dots, h_{2_j}, \nonumber \\
	& l_{1_1}, \dots, l_{1_{m_1}}, l_{2_1}, \dots, l_{2_{m_2}}, \nonumber \\
	& g_{1_1} \oplus_1 g_{2_1}, \dots, g_{1_n} \oplus_n g_{2_n} \} \nonumber
\end{align}
where $\oplus_i$, etc. denote the join condition.

For the example in Table~\ref{tab:example}, the local attributes are {\tt amn}
and {\tt rtg}, the aggregate attributes are {\tt cost} and {\tt duration},
and the join attributes are {\tt dst} and {\tt arr} for FlightsA, and {\tt src}
and {\tt dep} for FlightsB.

The \textsc{Aggregate Skyline Join Query (ASJQ)} is defined as:

\begin{definition}[Aggregate Skyline Join Queries (ASJQ)]
	The ASJQ queries retrieve the skyline set from the joined relation according
	to the preference functions of its $m_1 + m_2$ local and $n$ aggregate
	attributes.
\end{definition}

Dominance relationships between records can be defined based on the
\emph{attributes} on which a record dominates other records.  A tuple $r$ in
relation $R_i$ \emph{fully dominates} another tuple $s \in R_i$ if $r$ dominates
$s$ in both the local and the aggregate attributes of $R_i$.  If $r$ dominates
$s$ only in the local attributes, it is said to \emph{locally dominate} $s$.

The above definitions assume that whenever a tuple $t' = u \Join v'$ exists in
the final relation, the tuple $t = u \Join v$, where $v' \succ v$, also exists.
However, the join attributes of $v'$ and $v$ may be such that only $v'$
satisfies the join condition with $u$, but $v$ does not.  Consider flight 15 in
Table~\ref{tab:example}. It is dominated by flight 16 in the local attributes.
However, since they have different destinations, 15 can join with other flights
originating from $C$ (e.g., 23) which flight 16 cannot.  Hence, it must not be
considered to be dominated by flight 16.  In such cases, $t'$ may also exist as
a skyline in the final result as there is no $t$ to dominate it.  The problem is
that the local dominance did not take into account the join attributes.

In order to handle this, the join attributes must be taken into account when
full and local dominance relationships are defined.  Suppose, the join condition
that two join attributes $a$ from $A$ and $b$ from $B$ participate in is $A.a
\odot B.b$ where $\odot$ may be any of the following five comparison operators:
$=, <, \leq, >, \geq$ (we do not consider other operations in this paper).

\begin{table}[t]
	\centering
	\begin{tabular}{|c||c|c|}
		\hline
		Join & \multirow{2}{*}{$u \in A$ $\succ$ $u' \in A$ if} &
		\multirow{2}{*}{$v \in B$ $\succ$ $v' \in B$ if} \\
		condition & & \\
		\hline
		\hline
		$A.a = B.b$ & $u.a = u'.a$ & $v.b = v'.b$ \\
		$A.a < B.b$ & $u.a \leq u'.a$ & $v.b \geq v'.b$ \\
		$A.a \leq B.b$ & $u.a \leq u'.a$ & $v.b \geq v'.b$ \\
		$A.a > B.b$ & $u.a \geq u'.a$ & $v.b \leq v'.b$ \\
		$A.a \geq B.b$ & $u.a \geq u'.a$ & $v.b \leq v'.b$ \\
		\hline
	\end{tabular}
	\caption{Converting join conditions to skyline preferences.}
	\label{tab:joinsky}
\end{table}

Now, consider the tuple $u' \in A$. If it is dominated by tuple $u \in A$, then
it must be ensured that whenever $u'$ joins with $v \in B$, $u$ must also
satisfy the joining condition, i.e., if $u' \Join v$ is true, then $u \Join v$
must be true as well.  For example, if $\odot$ denotes $=$, then this translates
to $u.a = u'.a$ (both being equal to $v.b$); if $\odot$ denotes $<$, this
translates to $u.a < u'.a$, and similarly for the rest.  (The comparison
conditions are reversed for relation $B$.)  This condition can be incorporated
in the skyline finding routines as follows.

The join attribute is \emph{also} considered as a skyline attribute with the
preference function set appropriately as summarized in Table~\ref{tab:joinsky}.
This automatically ensures that whenever a tuple $u'$ is dominated by $u$, $u'$
can be discarded as the join of $u$ with $v$ can always be formed which will
ultimately dominate the join of $u'$ with $v$.

Based on the above discussion, the definitions of dominance relationships are
modified as follows.

\begin{definition}[Full dominance]
	A tuple $r$ in relation $R$ \emph{fully dominates} a tuple $s$ if $r$
	dominates $s$ in local, aggregate and join attributes of $R$.  
\end{definition}

\begin{definition}[Local dominance]
	A tuple $r$ in relation $R$ \emph{locally dominates} a tuple $s$ if $r$
	dominates $s$ in local and join attributes of $R$.  
\end{definition}

henceforth, whenever we mention local or aggregate attributes in the context of
dominance, we assume that the join attributes are incorporated within them.

Note that full dominance implies local dominance, but not vice versa.  The
corresponding definitions of \emph{full dominator} and \emph{local dominator}
are also specified.  Using these definitions, two kinds of skyline sets are also
defined.  A tuple $r$ in relation $R$ is in the \emph{full skyline} set if no
tuple in $R$ fully dominates $r$, and it is in the \emph{local skyline} set if
no tuple in $R$ locally dominates it.  A tuple that is in the local skyline set
is also in the full skyline set, but not vice versa.

\section{Related Work}
\label{sec:work}

The maximum vector problem or Pareto curve~\cite{321910} in the field of
computational geometry has been imported to databases forming the skyline
query~\cite{656550}. After the first skyline algorithm proposed by Kung et
al.~\cite{321910}, there were many algorithms devised by exploring the
properties of skylines. Some representative non-indexed algorithms are
SFS~\cite{1260846}, LESS~\cite{1083622}. Using index structures, algorithms such
as NN~\cite{1287394} and BBS~\cite{872814} have been proposed.

In~\cite{4233322}, Jin et al. proposed the multi relational skyline operator.
They also designed algorithms to find such skylines over multiple relations.
In~\cite{4354442}, Sun et al. coined the term ``skyline join'' in the context of
distributed environments.  They extended SaLSa~\cite{1183674} and also proposed
an iterative algorithm that prunes the search space in each step.  ASJQ queries
differ in that it extends the skyline join proposed in~\cite{4233322} with
aggregate operations performed during the join.  This renders the use of the
existing techniques inapplicable as they work only on the local attributes.

There are various algorithms for joining such as nested-loop join, indexed
nested-loop join, merge-join and hash-join~\cite{korth}.  Nested-loop joins can
be used regardless of the join condition. The other join techniques are more
efficient, but can handle only simple join conditions, such as natural joins or
equi-joins.  Any of these join algorithms that is applicable for the given query
can be used with ASJQ algorithms.

\section{Algorithms}
\label{sec:algo}

In this section, we describe the various algorithms that have been designed to
process the ASJQ queries.  We start with the na\"ive one before moving on to the
more sophisticated algorithm that uses the \emph{multiple skyline computations
(MSC)} approach.  The last two algorithms---\emph{dominator-based} and
\emph{iterative}---improves upon the MSC approach.  For each algorithm, we also
provide an analysis of its computation cost.

The pseudocode of the algorithms assume the procedures for
\emph{computeFullSkyline}, \emph{computeLocalSkyline}, \emph{computeJoin}, and
\emph{aggregate} methods.  The algorithms for these methods are not shown, since
any efficient skyline or join algorithm can be plugged into these methods.  The
aggregate method simply computes the aggregate operations on the specified
attributes.  Even though the efficiency of the entire method depends on the
complexities of these algorithms, we have not experimented with them as the
focus of this paper is on processing the ASJQ part.

\subsection{Na\"ive Algorithm}
\label{subsec:naive}

\begin{algorithm}[t]
\caption{Na\"ive Algorithm}
\label{alg:naive}
\begin{algorithmic}[1]
\REQUIRE Relations $A, B$, preferences $p$, aggregate operations $a$
\ENSURE Aggregate skyline join relation $S$
\STATE $J \leftarrow$ computeJoin$(A,B)$
\STATE $R \leftarrow$ Aggregate$(J,a)$
\STATE $S \leftarrow$ computeFullSkyline$(R,p)$
\end{algorithmic}
\end{algorithm}

The na\"ive method of processing ASJQ queries is shown in
Algorithm~\ref{alg:naive}.  It computes the join of the two input relations and
applies the aggregate operations, before computing the skyline on the joined and
aggregated relation using the preferences.  There are two costs involved in this
algorithm, \emph{joining} cost and cost for the \emph{skyline} computation. The
cost of \emph{aggregation} is not included, because it can be done when two
tuples are joined, without any extra cost.

\subsubsection{Analysis}

We denote the cost of a skyline operation on a relation of $N$ tuples having $a$
attributes by $S(N, a)$.  The cost of a join operation on two relations of size
$N_1$ and $N_2$ is denoted by $J(N_1, N_2)$.  Since the aggregate operations are
done as part of the join, the cost of those operations are not taken into
account separately.  Rather, if $g$ attributes are aggregated, the cost of the
join is denoted by $J(N_1, N_2, g)$, by incorporating the parameter within it.

Assuming the relations $A$ and $B$ contain $N_A$ and $N_B$ tuples respectively
with $n$ aggregate attributes, the cost of joining and aggregation is $J(N_A,
N_B, n)$.  The joined relation contains \emph{at most} $N_A N_B$ tuples, each
having $m_1 + m_2 + n$ attributes, and therefore, the cost of skyline operation
is $S(N_A N_B, m_1 + m_2 + n)$.

When operating on large relations, the above costs are impractical.
However, an advantage of the algorithm, apart from being the simplest to
implement, is the fact that it is independent of the distribution of the data.

\subsection{Performing Skylines before Join}
\label{subsec:joinattr}

Processing ASJQ queries can be made more efficient by pushing the join operation
after the full skylines have been evaluated in each relation, thereby discarding
tuples that are \emph{fully dominated} by other tuples.  These records are
guaranteed not to exist in the ASJQ result set.

Denoting the full skyline sets in each relation by $A_0$ and $B_0$ respectively,
and the non-skyline sets by $A'_0$ and $B'_0$ respectively, i.e., $A'_0 = A -
A_0$ and $B'_0 = B - B_0$, the following theorem shows that any tuple formed by
joining a tuple from either $A'_0$ or $B'_0$ or both cannot be part of the final
skyline set.

\begin{table}[t]
	\centering
	\begin{tabular}{|c|c|c|c|}
		\hline
		\multicolumn{3}{|c|}{Set} & Flight numbers \\
		\hline
		\hline
		\multirow{3}{*}{$A_0$} & \multicolumn{2}{c|}{$A_1$} & 11, 12 \\
		\cline{2-4}
		& \multirow{2}{*}{$A'_1$} & $A_2$ & 13, 14 \\
		\cline{3-4}
		& & $A'_2$ & 15, 16 \\
		\hline
		\multicolumn{3}{|c|}{$A'_0$} & 17 \\
		\hline
	\end{tabular}
	\caption{Categorization of relation FlightsA from Table~\ref{tab:example}.}
	\label{tab:example_cata}
\end{table}

\begin{theorem}
	\label{thrm:full}
	A tuple formed by joining a tuple that is not a full skyline in the
	individual relation never exists in the final skyline set.
\end{theorem}
\begin{proof}
	Consider a tuple $t' \in A_0 \Join B'_0$ formed by joining a tuple $u \in
	A_0$ with a tuple $v' \in B'_0$.  Since $v'$ is not a full skyline, there
	exists a tuple $v \in B_0$ that fully dominates $v'$.  Consider the tuple $t
	= u \Join v$.  The attributes in $l_1$ of $t$ are equal to those of $t'$,
	but dominate in $l_2$.  Consider an aggregate attribute $g'_i = g_{1_i}
	\oplus_i g'_{2_i}$ of $t'$.  The corresponding attribute value for $t$ is
	$g_i = g_{1_i} \oplus_i g_{2_i}$.  Since $g_{2_i}$ dominates $g'_{2_i}$ and
	$\oplus_i$ is a monotone aggregate function, $g_i$ dominates $g'_i$.  Hence,
	overall, the tuple $t$ dominates $t'$.  Thus, $t'$ cannot be part of the
	skyline.

	Similarly, any tuple in $A'_0 \Join B_0$ or $A'_0 \Join B'_0$ is dominated
	by tuples formed by joining the corresponding dominators, and will never
	exist in the final skyline set.
	\hfill{}
\end{proof}

As an example, consider flights 11 and 17.  Flight 11 fully dominates
flight 17 and satisfies the conditions for the join attributes as well.  This
ensures that any other flight joined with 17 (e.g., 21) can be joined with
11 as well, and the resulting joined tuple (11, 21) will surely dominate
(17, 21).  Hence, flight 17 need not be considered any further.  On the
other hand, even though flight 24 dominates flight 26 in the local and
aggregate attributes, the join attributes are not compatible as the sources of
the flights are different.  Hence, a tuple joined with 26 will not be
dominated by that joined by 24 as the latter tuple is invalid according to the
join criteria.

Thus, following the above theorem, the tuples from the sets $A'_0$ and $B'_0$ can be
discarded.  The remaining tuples may or may not exist in the final result set.
For example, consider flight 23 in the second relation.  It joins with three
tuples from the first relation as shown in Table~\ref{tab:example}(c).  Of
these, (11, 23) exists in the final skyline set while (13, 23) and
(15, 23) do not as they are dominated by (11, 23).

However, not all possible joined tuples from $A_0$ and $B_0$ need to be
examined.  Each full skyline set can be further divided by extracting the
\emph{local} skylines from them.  Suppose, the local skyline sets for $A_0$ and
$B_0$ be $A_1$ and $B_1$ respectively.  Correspondingly, let $A'_1$ and $B'_1$
be the set of non-skyline points within $A_0$ and $B_0$ respectively, i.e., they
are full skylines but not local skylines.  Mathematically, $A'_1 = A_0 - A_1$
and $B'_1 = B_0 - B_1$.  The following theorem shows that any tuple formed by
joining a tuple from either $A_1$ or $B_1$ or both \emph{must} be part of the
final skyline set.

\begin{table}[t]
	\centering
	\begin{tabular}{|c|c|c|c|}
		\hline
		\multicolumn{3}{|c|}{Set} & Flight numbers \\
		\hline
		\hline
		\multirow{3}{*}{$B_0$} & \multicolumn{2}{c|}{$B_1$} & 21, 22 \\
		\cline{2-4}
		& \multirow{2}{*}{$B'_1$} & $B_2$ & 23 \\
		\cline{3-4}
		& & $B'_2$ & 24, 25 \\
		\hline
		\multicolumn{3}{|c|}{$B'_0$} & 26, 27 \\
		\hline
	\end{tabular}
	\caption{Categorization of relation FlightsB from Table~\ref{tab:example}.}
	\label{tab:example_catb}
\end{table}

\begin{theorem}
	\label{thrm:local}
	The tuples formed by joining either or both of the tuples which are local
	skylines in the individual relations must exist in the final skyline set.
\end{theorem}
\begin{proof}
	Consider a tuple $t \in A_1 \Join B'_1$ formed by joining a tuple $u \in
	A_1$ with a tuple $v' \in B'_1$.  Since $u$ is a local skyline, there
	exists no tuple $u' \in A$ that locally (and therefore, fully) dominates
	$u$.  Thus, for any other joined tuple $t' \in A \Join B$, $t'$ cannot have
	local attributes of $A$ that dominate over $t$.  Thus, $t$ must be part of
	the skyline.

	Similarly, any tuple in $A'_1 \Join B_1$ or $A_1 \Join B_1$ is not dominated
	by any other tuple in all the attributes, and will therefore, always exist
	in the final skyline set.
	\hfill{}
\end{proof}

Consider flight 11 in the first relation and 21 in the second relation.
Both are local skylines in the corresponding full skyline sets, i.e., they are
part of $A_0$ and $B_0$ respectively.  Any tuple joined with 11 (e.g., 23)
must be part of the final skyline as no other tuple can dominate (11, 23) in
the local attributes of the first relation, i.e., {\tt f1.amn} and {\tt f1.rtg}.

However, nothing can be concluded directly about the tuples formed by joining
$A'_1$ with $B'_1$---they may or may not exist in the ASJQ result set.  Though
their local attributes will be dominated, their aggregate attributes may be
better, and therefore, they may be part of the skyline.  Consider the joined
tuple (13, 23).  It is dominated by (11, 21) even in the aggregate
attributes, and is, hence, not a skyline.  On the other hand, the tuple (14,
24) is a skyline, even though 14 is locally dominated by 11 and 24 by
21; however, the aggregate attributes of (14, 24) are more preferable.  Hence,
the tuples in $A_1' \Join B_1'$ needs to be processed to determine the ASJQ
records in it.

The ASJQ algorithms utilize Theorem~\ref{thrm:full} and Theorem~\ref{thrm:local}
to reduce the processing by first determining the skyline sets before joining.

In addition to the high processing costs, the na\"ive algorithm suffers from the
problem of non-progressive result generation, i.e., it presents the
results only after complete processing of the algorithms.  In real applications
with large datasets, query processing may take a lot of time, and this large
response time, even for the first result, may be undesirable for many users.
This can be handled by devising \emph{online} algorithms that generate a subset
of the full results quickly and progressively generates the tuples thereafter.
Though the full results are still output only after complete processing, these
can be used in real-time applications.

MSC and the next set of algorithms achieve this by generating tuples that are
sure to be in the final skyline set \emph{without} processing all the tuples in
the joined relation.  

\subsection{Multiple Skyline Computations (MSC) Algorithm}
\label{subsec:msc}

\begin{algorithm}[t]
\caption{MSC Algorithm}
\label{alg:msc}
\begin{algorithmic}[1]
\REQUIRE Relations $A, B$, preferences $p$, aggregate operations $a$
\ENSURE Aggregate skyline join relation $S$
\STATE $A_0 \leftarrow$ computeFullSkyline$(A)$
\STATE $B_0 \leftarrow$ computeFullSkyline$(B)$
\STATE $(A_1,A_1') \leftarrow$ computeLocalSkyline$(A_0)$
\STATE $(B_1,B_1') \leftarrow$ computeLocalSkyline$(B_0)$
\STATE $J \leftarrow$ computeJoin$(A_1,B_1)$ $\cup$ computeJoin$(A_1,B_1')$ $\cup$ computeJoin$(A_1',B_1)$
\STATE $R \leftarrow$ Aggregate$(J,a)$
\STATE $J' \leftarrow$ computeJoin$(A_1',B_1')$
\STATE $R' \leftarrow$ Aggregate$(J',a)$
\STATE $S \leftarrow R \ \cup$ computeFullSkyline$(R',R)$
\quad /* finds skyline points in $R'$ by treating the current skyline as $R$ */
\end{algorithmic}
\end{algorithm}

The Multiple Skyline Computations (MSC) algorithm uses the results of the above
two theorems, and immediately outputs the tuples in $A_1 \Join B_1$, $A_1 \Join
B_1'$, and $A_1' \Join B_1$.  It then examines the tuples from $A_1' \Join B_1'$
to determine whether they are part of the final skyline set.
Algorithm~\ref{alg:msc} shows the complete algorithm.

Moreover, processing the joined relation, which is generally large,
constitutes most of the processing cost.  Hence, algorithms that reduce the
number of comparisons in the joined relation without processing the whole
relation improves the efficiency of ASJQ processing.

Table~\ref{tab:example_cata} and Table~\ref{tab:example_catb} respectively show
the division of the sets $A$ and $B$ from Table~\ref{tab:example} into the
different categories.  The na\"ive algorithm finds the skyline by examining $11$
joined tuples.  Theorem~\ref{thrm:full} reduces the number of joined tuples to
$6$ (as shown in Table~\ref{tab:example}(c)).  By applying
Theorem~\ref{thrm:local}, the MSC algorithm reduces it further by computing the
sets $A'_1$ and $B'_1$.  The total number of tuples in $A'_1 \Join B'_1$ on
which the final skyline needs to be computed is only $3$.

\subsubsection{Analysis}

We next analyze the costs of the MSC algorithm.  Using the same notation as in
the analysis of the na\"ive algorithm, the first two full skyline computations
has a cost of $S(N_A, j + m_1 + n) + S(N_B, j + m_2 + n)$, where $n_C$ denotes
the cardinality of the set $C$.  The cost of computing the local skylines next
are $S(N_{A_0}, m_1) + S(N_{B_0}, m_2)$.

The total cost of computing the three joins, $A_1 \Join B_1$, $A_1 \Join B_1'$,
and $A_1' \Join B_1$, is $J(A_1, B_1, n) + J(A_1, B'_1, n) + J(A'_1, B_1, n)$.
The full skyline operator is applied on the tuples from $A'_1 \Join B'_1$,
thereby incurring a cost of at most $S(N_{A'_1} . N_{B'_1}, m_1 + m_2 + n)$.

The MSC algorithm performs significantly better than the na\"ive one when the
cardinality of the full skyline set is low but that of the local skyline sets is
high.  A number of skyline tuples can be generated quickly and only a few ones
(those in $A'_1 \Join B'_1$) require a complete investigation.  Since the
skylines are computed locally, the number of local attributes plays a big role.
With more number of local attributes, the size of $A_1$ ($B_1$) grows.  However,
in that case, the cardinality of $A_0$ ($B_0$), and hence, that of $A'_1$
($B'_1$) will be large as well, thereby reducing some of the benefits of the MSC
algorithm.  Section~\ref{sec:exp} analyzes the effect of these parameters.

\begin{figure}[t]
\centering
\includegraphics[width=0.60\columnwidth]{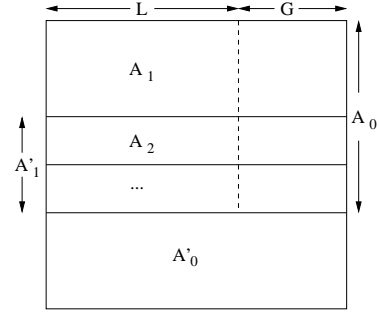}
\caption{Break-up of skyline sets for iterative algorithm.}
\label{fig:iter}
\end{figure}

\subsection{Dominator-Based Approach}
\label{subsec:dom}

In order to further reduce the processing cost of tuples from $A'_1 \Join B'_1$,
the following two algorithms are designed.  The first algorithm makes use of
\emph{dominator} properties among the tuples and prunes away unnecessary
comparisons while determining the ASJQ records within the set $A_1' \Join B_1'$.

Consider a tuple $t' \in A_1' \Join B_1'$ formed by joining tuples $u' \in A'_1$
and $v' \in B'_1$, i.e., $t' = u' \Join v'$.  The tuple $t'$ can be dominated by
only certain records of the skyline set $(A_1 \Join B_1) \cup (A'_1 \Join B_1)
\cup (A_1 \Join B'_1)$.  Identifying these records avoids comparing with the
whole sets.  Suppose, the local dominators of $u'$ ($v'$) are represented by
$ld(u')$ ($ld(v')$).  The following lemma proves that $t'$ can be dominated by
only those tuples $t$ that are in $ld(u') \Join ld(v')$, and nothing else.

\begin{lemma}
	\label{lem:dom}
	A tuple $t' = u' \Join v'$ in $A'_1 \Join B'_1$ can be dominated by only those
	tuples that are formed by joining tuples in the local dominator sets of $u'$
	and $v'$, i.e., in $ld(u') \Join ld(v')$. 
\end{lemma}
\begin{proof}
	Consider a tuple $t' = u' \Join v' \in A'_1 \Join B'_1$.  Also, consider $u$
	which is \emph{not} a local
	dominator of $u'$, i.e., $u \notin ld(u')$, and a tuple $t$ formed by
	joining $u$ with any $v \in B$.  The local attributes $l_1$ of
	$t'$ are not dominated by those in $t$ as then $u'$ would have been
	dominated by $u$.  Thus, $t$ cannot dominate $t'$.
	Similarly, any $t$ formed by joining any $u$ with $v \notin ld(v')$ cannot
	dominate $t'$ as the local attributes of the second relation will not be
	dominated.  Hence, if $t'$ can only be dominated by $t \in ld(u') \Join
	ld(v')$.
	\hfill{}
\end{proof}

\begin{algorithm}[t]
\caption{Dominator-Based Algorithm}
\label{alg:dom}
\begin{algorithmic}[1]
\REQUIRE Relations $A, B$, local preferences $l$, preferences $p$, aggregate operations $a$
\ENSURE Aggregate skyline join relation $S$
\STATE $A_0 \leftarrow$ computeFullSkyline$(A)$
\STATE $B_0 \leftarrow$ computeFullSkyline$(B)$
\STATE $(A_1,A_1') \leftarrow$ computeLocalSkyline$(A_0)$
\STATE $(B_1,B_1') \leftarrow$ computeLocalSkyline$(B_0)$
\STATE $(A_1,A_1',D_A) \leftarrow$ findLocalDominators$(A_0,l)$
\quad /* using Algorithm~\ref{alg:findingdom} */
\STATE $(B_1,B_1',D_B) \leftarrow$ findLocalDominators$(B_0,l)$
\quad /* using Algorithm~\ref{alg:findingdom} */
\STATE $J \leftarrow$ computeJoin$(A_1,B_1)$ $\cup$ computeJoin$(A_1,B_1')$ $\cup$ computeJoin$(A_1',B_1)$
\STATE $R \leftarrow$ Aggregate$(J,a)$
\STATE $J' \leftarrow$ computeJoin$(A_1',B_1')$
\STATE $R' \leftarrow$ Aggregate$(J',a)$
\STATE $S \leftarrow R \ \cup$ computeSkylineUsingDominators$(R',D_A,D_B)$
\quad /* finds skyline points in $R'$ by using dominator sets $D_A,D_B$ (Algorithm~\ref{alg:skylinedom}) */
\end{algorithmic}
\end{algorithm}

The records in $ld(u') \Join ld(v')$ are not guaranteed to dominate $t'$ though.
This is due to the fact that $u'$ contains aggregate attributes that are not
dominated by those of $ld(u')$ (the reason being $u'$ belonging to the set
$A_0$, i.e., it is a full skyline).  Hence, the tuple $t'$ may need to be
compared with all the tuples in $ld(u') \Join ld(v')$.  This reduces the
computation cost of the last step of the MSC algorithm significantly as it is
not compared with all tuples of $(A_1 \Join B_1) \cup (A'_1 \Join B_1) \cup (A_1
\Join B'_1)$.

However, the previous steps perform more work by finding the dominator sets for
each tuple not in the local skyline set.  In other words, by increasing the cost
of the MSC step to draw some conclusions among the records, the overall
computational cost is reduced by utilizing those properties in the latter stages
of the algorithm.

\begin{algorithm}[t]
\caption{Skyline Computation and Finding Dominators}
\label{alg:findingdom}
\begin{algorithmic}[1]
\REQUIRE Relation $A_0$, local preferences $p$
\ENSURE Skyline set $A_1$, Non-skyline set $A_1'$, Dominator set $D_1$
\WHILE{$r' \leftarrow$ readRecord$(A_0)$}
\STATE $flag \leftarrow 0$
\WHILE{$r \leftarrow$ readRecord$(A_0)$}
\IF{$r$ locally dominates $r'$ using preferences $p$}
\STATE $D(r') \leftarrow D(r') \cup r$
\STATE $flag \leftarrow 1$
\ENDIF
\ENDWHILE
\IF{flag $= 0$}
\STATE $A_1 \leftarrow A_1 \cup r'$
\ELSE
\STATE $A_1' \leftarrow A_1' \cup r'$
\STATE $D_1 \leftarrow D_1 \cup D(r')$
\ENDIF
\ENDWHILE
\STATE $S \leftarrow (A_1, A'_1, D_1)$
\end{algorithmic}
\end{algorithm}

Algorithm~\ref{alg:dom} summarizes the approach.  It uses
Algorithm~\ref{alg:findingdom} to find the local dominator sets for each record
that is in $A'_1$ (and $B'_1$).  Algorithm~\ref{alg:skylinedom} shows the
subroutine that utilizes these local dominator sets to determine whether a tuple
is in the final skyline set.

In the example in Table~\ref{tab:example}, flight 13 is locally dominated by
flights 11 and 12 while flight 23 is locally dominated by flights 21 and 22.
Therefore, to determine whether tuple (13, 23) is a skyline in the ASJQ set, it
needs to be checked only against (11, 21).  (The other combinations do not
generate valid joined tuples.) This is a large improvement as opposed to the MSC
algorithm that checks (13, 23) against $5$ joined tuples from $(A_1 \Join B_1)
\cup (A'_1 \Join B_1) \cup (A_1 \Join B'_1)$.

\begin{algorithm}[t]
\caption{Skyline Computation using Dominators}
\label{alg:skylinedom}
\begin{algorithmic}[1]
\REQUIRE Non-skyline set $R'$, Dominator sets $D_A, D_B$,
preferences $p$, aggregate operations $a$
\ENSURE Skyline set $R$
\WHILE{$r' \leftarrow$ readRecord$(R')$}
\STATE $r' \leftarrow u \Join v$
\STATE flag $\leftarrow 0$
\WHILE{$d_A \leftarrow$ readDominator$(r',D_A)$}
\STATE /* read record from $D_A$ that locally dominates $u$ */
\WHILE{$d_B \leftarrow$ readDominator$(r',D_B)$}
\STATE /* read record from $D_B$ that locally dominates $v$ */
\STATE $r \leftarrow$ Aggregate$(d_A \Join d_B, a)$
\quad /* read full record from $R$ that has $d_A$ and $d_B$ */
\IF{$r$ fully dominates $r'$ according to preferences $p$}
\STATE discard $r'$
\STATE flag $\leftarrow 1$
\STATE \textbf{break}
\ENDIF
\ENDWHILE
\IF{flag $= 1$}
\STATE \textbf{break}
\ENDIF
\ENDWHILE
\IF{flag $= 0$}
\STATE $R \leftarrow R \cup r'$
\ENDIF
\ENDWHILE
\end{algorithmic}
\end{algorithm}

\subsubsection{Analysis}

We now analyze the costs of the dominator-based algorithm with respect to the
MSC algorithm.  The first two full skyline computations has the same cost of
$S(N_A, m_1 + n) + S(N_B, m_2 + n)$.  The local skylines are computed next
having a total cost of $S(N_{A_0}, m_1) + S(N_{B_0}, m_2)$.

In addition to the skyline computations, the dominator sets are computed.
Denoting the cost of dominator computation by $D$, the cost is $D(N_{A_0}) +
D(N_{B_0})$.  Note that even though the dominators for only $A'_1$ and $B'_1$
tuples are maintained, all the tuples of $A_0$ and $B_0$ need to be processed.
Suppose, the size of the dominator sets are $d_{A'_1}$ and $d_{B'_1}$
respectively.

The skyline operator is next applied on the tuples from $A'_1 \Join B'_1$ using
the dominators found in the previous step.  This cost is at most $SD(N_{A'_1} .
N_{B'_1}, d_{A'_1} \times d_{B'_1}, n)$.  Note that the dimensionality of the
skyline operation using dominators here is only $n$, i.e., only the aggregate
attributes need to be checked for dominance, as the local attributes are, by
definition, dominated by the local dominators.

Finally, the total cost of computing the three other joins, $A_1 \Join B_1$,
$A_1 \Join B_1'$, and $A_1' \Join B_1$, is the same as that of the MSC
algorithm, and can be denoted by $J(A_1, B_1, n) + J(A_1, B'_1, n) + J(A'_1,
B_1, n)$.

The dominator-based algorithm thus performs well when the dominator sets are
small.  Otherwise, the overhead of dominator computation may be too large to
gain any speedup over the MSC algorithm.  Section~\ref{sec:exp} compares these
algorithms experimentally.

\subsection{Iterative Algorithm}
\label{subsec:iter}

The dominator-based algorithm involves computation of local dominator sets which
can be costly.  By eliminating the costly dominator computations, we devise
another algorithm which is iterative in nature and is an attractive online
algorithm.

The main cost of the MSC algorithm is the skyline computation on the join of the
two sets $A'_1$ and $B'_1$.  This algorithm reduces the complexity of this cost
by further dividing the set $A'_1$ ($B'_1$) into local skylines $A_2$ ($B_2$)
and non-skylines $A'_2$ ($B'_2$).  Iteratively, this is proceeded until the
cardinality of the non-skyline set is less than a preset threshold $\delta$.
The relation $A_0$ (similarly, $B_0$) is thus subdivided into $A_1, A_2, \dots,
A_k, A'_k$, as shown in Figure~\ref{fig:iter}.

By observing certain relationships among these sets, we can determine that the
dominators of the records of a set exist only in a few of the other sets, and it
needs to be compared only with those sets.  For example, a tuple in $A_2 \Join
B_2$ needs to be compared with tuples in $A_1 \Join B_1$ only, thereby
eliminating unnecessary comparisons with tuples in $(A_1 \Join B_1') \cup (A_1'
\Join B_1) \cup (A_1' \Join B_1')$.

\begin{lemma}
	\label{lem:target}
	A tuple in $A_2 \Join B_2$ can be dominated only by a tuple in $A_1 \Join
	B_1$ and not by any tuple in $(A_1 \Join B_1') \cup (A_1' \Join B_1) \cup
	(A_1' \Join B_1')$.
\end{lemma}
\begin{proof}
	Consider a tuple $t' = u' \Join v' \in A_2 \Join B_2$.  Consider any other
	tuple $t = u \Join v \in A'_1 \Join B_1$.  If $t$ dominates $t'$, then the
	$l_1$ local attributes of $t$ pertaining to $u$ must dominate that of $u'$.
	However, since $A_2$ is in the local skyline set of $A'_1$, this contradicts
	the fact that no tuple in $A'_1$ locally dominates a tuple in $A_2$.
	Similarly, no tuple in $A_1 \Join B'_1$ or $A'_1 \Join B'_1$ can dominate
	$t'$.
	\hfill{}
\end{proof}

For each such set $A_i \Join B_j$, there exists \emph{target sets}, within which
it has to search for its dominators and test for the ASJQ requisites.  We show
the target sets up to two iterations in Table~\ref{tab:iter_target}.

\begin{algorithm}[t]
\caption{Iterative Algorithm}
\label{alg:iter}
\begin{algorithmic}[1]
\REQUIRE Relations $A, B$, local preferences $l$, preferences $p$, aggregate operations $a$
\ENSURE Aggregate skyline join relation $S$
\STATE $A_0 \leftarrow$ computeFullSkyline$(A)$
\STATE $B_0 \leftarrow$ computeFullSkyline$(B)$
\STATE $i \leftarrow 1$
\WHILE{$|A'_i| \leq \delta$}
\STATE $A_{i+1} \leftarrow$ computeLocalSkyline$(A_i)$
\STATE $i \leftarrow i+1$
\ENDWHILE
\STATE $L_A \leftarrow i$
\quad /* Number of levels of skyline sets in $A$ */
\STATE $j \leftarrow 1$
\WHILE{$|B'_j| \leq \delta$}
\STATE $B_{j+1} \leftarrow$ computeLocalSkyline$(B_j)$
\STATE $j \leftarrow j+1$
\ENDWHILE
\STATE $L_B \leftarrow j$
\quad /* Number of levels of skyline sets in $B$ */
\STATE $J \leftarrow$ computeJoin$(A_1,B_1)$ $\cup$ computeJoin$(A_1,B_1')$ $\cup$ computeJoin$(A_1',B_1)$
\STATE $G \leftarrow$ Aggregate$(J,a)$
\STATE $S \leftarrow G$
\STATE $i \leftarrow 1,\ j \leftarrow 1$
\WHILE{$ i \leq L_A$}
\WHILE{$ j \leq L_B$}
\STATE $J_{ij}' \leftarrow$ computeJoin$(A'_i,B'_j)$
\STATE $G_{ij}' \leftarrow$ Aggregate$(J_{ij}',a)$
\STATE $S \leftarrow S \ \cup$ computeSkylineUsingTargetSets$(G'_{ij})$
\STATE $j \leftarrow j+1$
\ENDWHILE
\STATE $i \leftarrow i+1$
\ENDWHILE
\end{algorithmic}
\end{algorithm}

The iterative algorithm is summarized in Algorithm~\ref{alg:iter}.  In each
relation, the skyline sets are computed till the threshold $\delta$ is reached.
All combinations of such non-skyline sets are then joined, and the dominators
for aggregates are checked only against the corresponding target sets.

\begin{table}[t]
\centering
\begin{tabular}{|c|l|} \hline
	Set & \multicolumn{1}{c|}{Target Sets} \\
\hline
\hline
$A_2\Join B_2$& $A_1\Join B_1$ \\ \hline
$A_2\Join B_2'$& $A_1\Join B_1$, $A_1\Join B_1'$ \\ \hline
$A_2'\Join B_2$& $A_1\Join B_1$, $A_1'\Join B_1$ \\ \hline
$A_2'\Join B_2'$& $A_1\Join B_1$, $A_1\Join B_1'$, $A_1'\Join B_1$\\ \hline
\end{tabular}
\caption{Target sets for iterative algorithm.}
\label{tab:iter_target}
\end{table}

The computeSkylineUsingTargetSets method mentioned in the algorithm determines
the skyline records in the set $S_{ij}$ by comparing only with the target sets
corresponding to it as shown in Table~\ref{tab:iter_target}.  The first
iteration of the iterative algorithm remains the same as in the MSC algorithm.
In the second iteration, the sets $A_2$ and $B_2$ are joined and these are
compared with only the target sets shown in Table~\ref{tab:iter_target}.
Similarly, in the next iteration, local skyline is further computed in $A_2'$
and $B_2'$, and so on until the cardinality falls below the threshold $\delta$.

For the running example given in Table~\ref{tab:example}, the break-up of the
relations into the different sets $A_1, A_2$, etc. are shown in
Table~\ref{tab:example_cata} and Table~\ref{tab:example_catb}.  Here, $A_2'$ and
$B_2'$ are not further categorized, as they have only two tuples, and no tuple
dominates the other.  In other words, $A_3=A_2'$ and $B_3=B_2'$, and the sets
$A'_3$ and $B'_3$ are empty.  Hence, this is considered as the last iteration.

\subsubsection{Analysis}

The cost analysis of the iterative algorithm depends heavily on the cardinality
of the non-skyline sets produced progressively.  The number of tuples that are
joined remains the same as in the MSC approach.  However, the ASJQ
computation cost for the tuples in $A_i'\Join B_j'$ reduces significantly,
since the search space for each tuple is iteratively pruned, and is thus,
optimized.

As a result, it performs significantly better in comparison to the other
algorithms for datasets with large non-(full)skyline sets.  This is due to the
fact that the non-skyline sets are not blindly joined with each other, but
rather only the relevant records are joined and compared in a progressive
manner.  This cuts down many unnecessary skyline tests, thereby improving the
efficiency.

\subsection{ASJQ with Single Aggregate Attribute}
\label{subsec:aggr}

A special case of the \emph{Aggregate Skyline Join Query} is when it involves
only a single aggregate attribute.  The processing then becomes substantially
easier.  As shown in Section~\ref{subsec:joinattr}, the records which do not
exist in the \emph{full skyline} set of each relation (i.e., those in $A'_0$ and
$B'_0$) are discarded.  However, when the number of aggregate attributes is one,
even the tuples formed by joining $A_0$ with $B_0$ do not need to be examined.
An interesting observation, summarized in the following lemma, leads to the
expeditious generation of the skyline points.  The tuples in $A_0 \Join B_0$ are
guaranteed to be part of the final skyline set.

\begin{lemma}
	\label{lem:aggr}
	When there is only one aggregate attribute, the tuples formed by joining the
	full skyline points of each relation always exist in the ASJQ result set.
\end{lemma}
\begin{proof}
	Consider the set $A_0$ ($B_0$) to be divided it into local skyline set $A_1$
	($B_1$) and non-skyline records $A'_1$ ($B'_1$).  Using
	Theorem~\ref{thrm:local}, the tuples in $A_1 \Join B_1$ exist in the final
	skyline set.
	
	Consider a tuple $t' = u' \Join v' \in A_1' \Join B_1'$.  We claim that
	there does not exist any tuple $t = u \Join v$ that dominates $t'$ fully.
	To counter the claim, assume that such a tuple $t$ exists.  Since $t$
	dominates $t'$, the local attributes of $t$ must dominate those in $t'$.
	Thus, $u \succ u'$ and $v \succ v'$.  Next, consider the aggregate attribute
	of $t'$, expressed as $g_{t'} = g_{u'} \oplus g_{v'}$.  Note that since $u'$
	is a full skyline record, no tuple and in particular $u$, can dominate $u'$
	in all the attributes.  That is to say, $u'$ must dominate $u$ in the
	aggregate attribute, since it is being dominated in all the other (local)
	attributes, i.e., $g_{u'}$ dominates $g_u$.  Similarly, $g_{v'}$ dominates
	$g_v$.  Since the aggregate function $\oplus$ is a \emph{monotone} function,
	$g_{t'} = g_{u'} \oplus g_{v'}$ dominates $g_t = g_u \oplus g_v$.
	Therefore, the claim that $t$ dominates $t'$ fully is false.  Consequently,
	the tuple $t'$ must be in the final skyline set.

	Similarly, any tuple in $(A'_1 \Join B_1) \cup (A_1 \Join B'_1)$ must also
	be a skyline record.  Together, all the tuples in $A_0 \Join B_0$ exist in
	the ASJQ result set.
	\hfill{}
\end{proof}

Therefore, when there is only one aggregate attribute, an algorithm that divides
the full skyline sets into local skylines and non-skylines, and returns the join
of the two local skyline sets as the final ASJQ result, is the \emph{optimal}
algorithm.

\section{Experimental Evaluation}
\label{sec:exp}

\begin{table}[t]
\centering
\begin{tabular}{|c|c|c|}
\hline
Parameter & Symbol & Value \\
\hline
\hline
Number of local attributes & $L$ & 2 \\
Number of aggregate attributes & $G$ & 2 \\
Cardinality of datasets & $N$ & 40000 \\
Number of categories & $C$ & 10 \\
Distribution of datasets & $D$ & Correlated \\
\hline
\end{tabular}
\caption{Default parameters for synthetic data.}
\label{tab:default}
\end{table}

In this section, we evaluate the ASJQ algorithms experimentally.  We implemented
them in Java on an Intel Core2Duo 2GHz machine with 2GB RAM in Linux
environment.  We used the synthetic dataset generator given in
\url{http://www.pgfoundry.org/projects/randdataset/} and used in~\cite{656550}.
We also used a real dataset of statistics of basketball players obtained from
\url{http://www.databasebasketball.com/}.  For the skyline algorithm, we
employed the SFS method~\cite{1260846}\footnote{The choice of SFS versus other
algorithms such as LESS~\cite{1083622} does not matter as the focus is on the
join and not the skyline computation.}, and used hash-join~\cite{korth} for
implementing the join.

We analyze the execution times of the four algorithms: (1)~Na\"ive, (2)~MSC,
(3)~Dominator-based, and (4)~Iterative, based on the following parameters:
(i)~number of local attributes ($L$), (ii)~number of aggregate attributes ($G$),
(iii)~cardinality of datasets ($N$), (iv)~number of categories in each relation
for joining attribute assuming equi-join ($C$), and (v)~distribution of datasets
($D$).  Unless mentioned otherwise, the default settings of the five parameters
for experiments with the synthetic data are given in Table~\ref{tab:default}.

\begin{table}[t]
\centering
\begin{tabular}{|c|c|c|c|c|c|}
\hline
 & $N$ & $D$ & $L$ & $G$ & $C$ \\
\hline
\hline
Setting 1 &10000&Correlated&2&3&10 \\
Setting 2 &10000&Correlated&3&2&10 \\
Setting 3 &3162&Independent&2&2&10 \\
Setting 4 &316&Anti-Correlated&1&2&10 \\
Setting 5 &316&Independent&2&1&10 \\
\hline
\end{tabular}
\caption{Experimental settings.}
\label{tab:settings}
\end{table}

\subsection{Performance of the na\"ive algorithm}

The first experiment examines the difference in performance of the na\"ive with
the other ASJQ algorithms.  We use five random settings of synthetic datasets as
shown in Table~\ref{tab:settings}.  The plots in Figure~\ref{fig:settings}
compare the execution times of the different algorithms.  The join condition is
an equi-join on a single attribute.

For all the five settings, the na\"ive algorithm requires much higher running
times.  Further, while the performance of the other algorithms depends on the
final cardinality of the ASJQ result set and is proportional to it, the na\"ive
algorithm is more or less independent of the final cardinality.  This is due to
the fact that it spends most of the time in computing the join of the relations
and then applies the skyline operator on the large joined relation.

Due to the large gap in the running times, we conclude that the na\"ive
algorithm is not practical in comparison to the other algorithms.  Consequently,
we do not report the results of the na\"ive algorithm any further.

\begin{figure}[t]
\begin{center}
\begin{tabular}{cc}
	\hspace*{-10mm}
	\includegraphics[width=\figwidth]{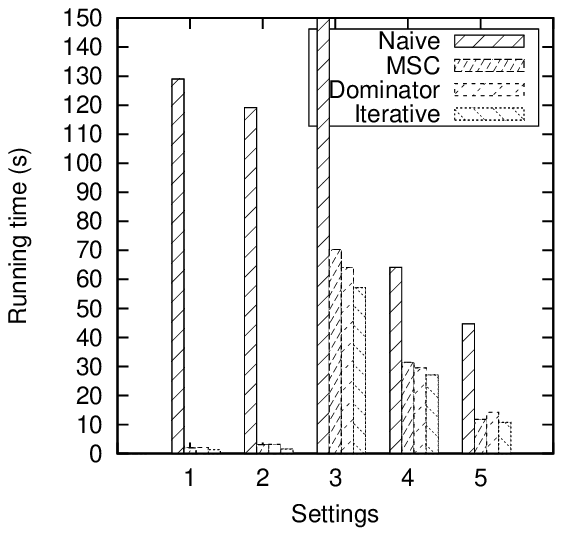} &
	\hspace*{-18mm}
	\includegraphics[width=\figwidth]{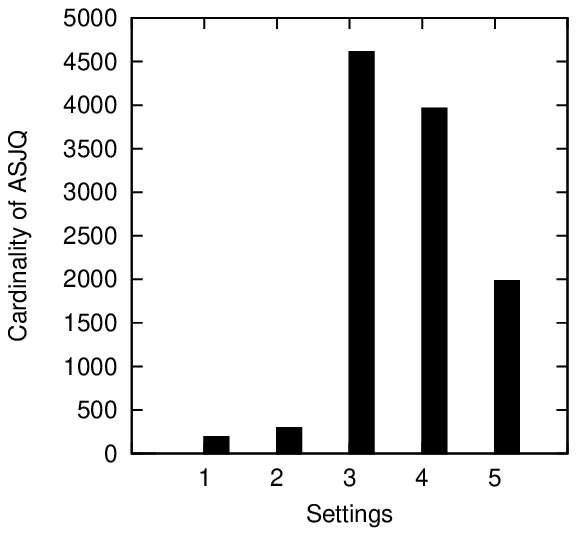} \\
	\hspace*{-10mm}
(a) Running time. &
	\hspace*{-18mm}
(b) Cardinality of ASJQ.
\end{tabular}
\end{center}
\caption{Comparison with na\"ive algorithm.}
\label{fig:settings}
\end{figure}

\subsection{Effect of dimensionality}
\label{subsec:dim_exp}

The first experiment measures the effect of the number of local attributes ($L$)
on the algorithms.  Figure~\ref{fig:local}(a) shows that the running time
increases sharply when $L$ increases.  This can be attributed to the fact that
the cardinality of the ASJQ result set increases almost exponentially
(Figure~\ref{fig:local}(b)).  As the dimensionality of the datasets (i.e., the
number of attributes) increases, the probability of a tuple being dominated in
all the attributes decreases, thereby sharply increasing the number of skyline
records.

The iterative algorithm shows the best scalability since it processes the
skyline sets progressively.  At lower dimensions, the time to find the full
skyline sets in the individual relations is the dominating factor of the overall
time, and hence, there is little difference between the algorithms.

\begin{figure}[t]
\begin{center}
\begin{tabular}{cc}
	\hspace*{-10mm}
\includegraphics[width=\figwidth]{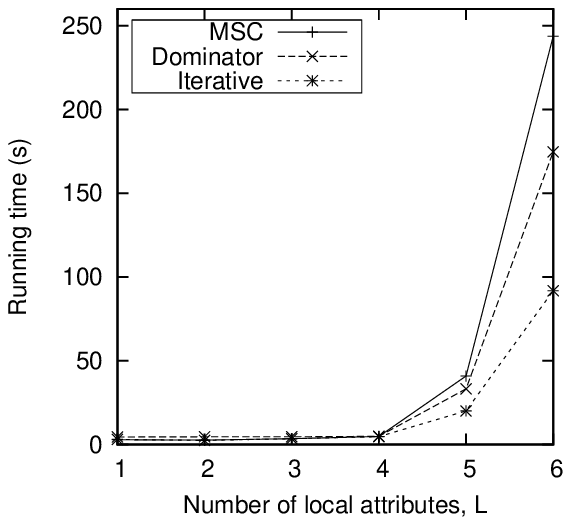} &
	\hspace*{-18mm}
\includegraphics[width=\figwidth]{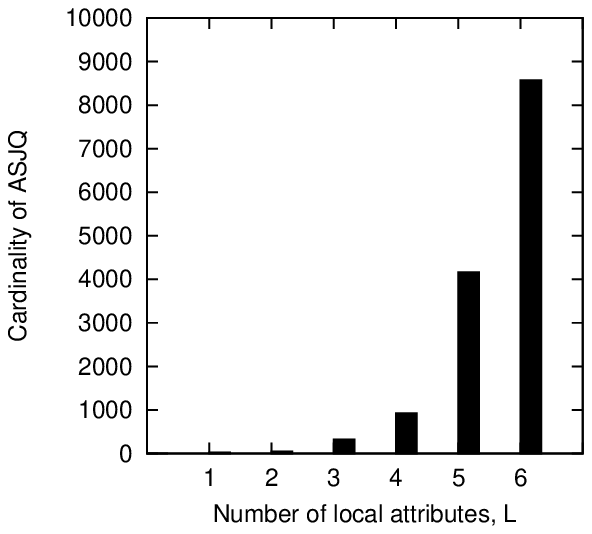} \\
	\hspace*{-10mm}
(a) Running time. &
	\hspace*{-18mm}
(b) Cardinality of ASJQ.
\end{tabular}
\end{center}
\caption{Effect of number of local attributes.}
\label{fig:local}
\end{figure}

Figure~\ref{fig:aggr}(a) and Figure~\ref{fig:aggr}(b) show similar trends.
Interestingly, the absolute times are much lower than the corresponding number
of local attributes.  Incrementing the number of local attributes increases the
dimensionality in the joined relation by two, whereas it only increases by one
for aggregate attributes.  Thus, the effect of dimensionality is less
pronounced.  Consequently, the cardinality of the final ASJQ set is less.

The MSC algorithm performs better than the dominator-based algorithm since the
number of local attributes is small and the local dominator sets are larger.
Consequently, the overhead of dominator computation and comparison offsets the
advantages.

\begin{figure}[t]
\begin{center}
\begin{tabular}{cc}
	\hspace*{-10mm}
\includegraphics[width=\figwidth]{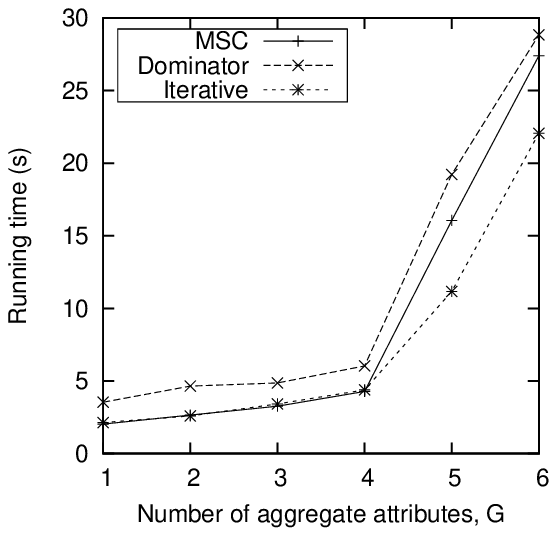} &
	\hspace*{-18mm}
\includegraphics[width=\figwidth]{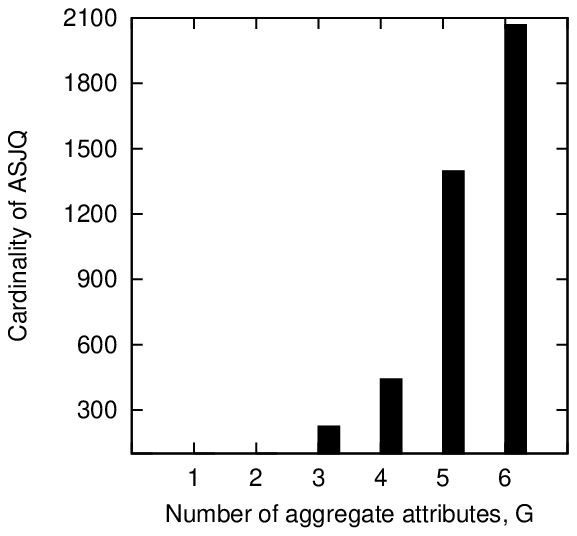} \\
	\hspace*{-10mm}
(a) Running time. &
	\hspace*{-18mm}
(b) Cardinality of ASJQ.
\end{tabular}
\end{center}
\caption{Effect of number of aggregate attributes.}
\label{fig:aggr}
\end{figure}

\subsection{Effect of dataset cardinality}
\label{subsec:card_exp}

The next experiment measures the effect of the cardinality of the individual
relations on ASJQ processing.  The cardinality of the joined relation increases
\emph{quadratically} with the individual cardinality, assuming that the data
distribution remains the same.  For example, assume two datasets with $N = 10^4$
tuples each.  If an equi-join condition is used where the number of categories
of the joining attributes is assumed to be $10$, each category has on an average
$10^3$ tuples.  Hence, the total cardinality of the joined relation becomes $10
\times (10^3)^2 = 10^7$.

Figure~\ref{fig:card}, however, shows that the cardinality of the ASJQ result
set does not increase quadratically.  (The figure reports results for $4$ local
and $4$ aggregate attributes.  The cardinality and the running time for $L=2$
and $G=2$ were too low.)  The number of skyline records depends more on other
parameters of the dataset, such as dimensionality and distribution.
Consequently, the scalability of the ASJQ algorithms with $N$ is better.

\begin{figure}[t]
\begin{center}
\begin{tabular}{cc}
	\hspace*{-10mm}
\includegraphics[width=\figwidth]{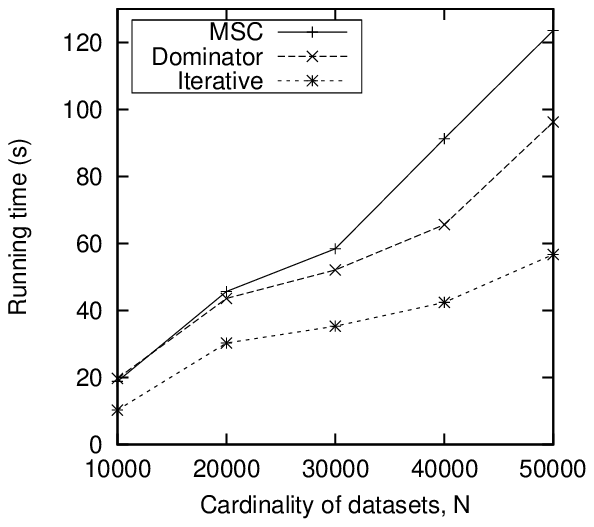} &
	\hspace*{-18mm}
\includegraphics[width=\figwidth]{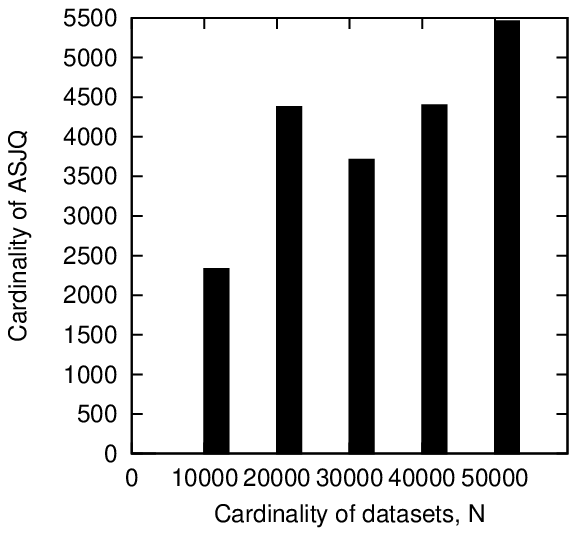} \\
	\hspace*{-10mm}
(a) Running time. &
	\hspace*{-18mm}
(b) Cardinality of ASJQ.
\end{tabular}
\end{center}
\caption{Effect of dataset cardinality.}
\label{fig:card}
\end{figure}

\subsection{Effect of dataset distribution}
\label{subsec:dist_exp}

\begin{figure}[t]
\begin{center}
\begin{tabular}{cc}
	\hspace*{-10mm}
\includegraphics[width=\figwidth]{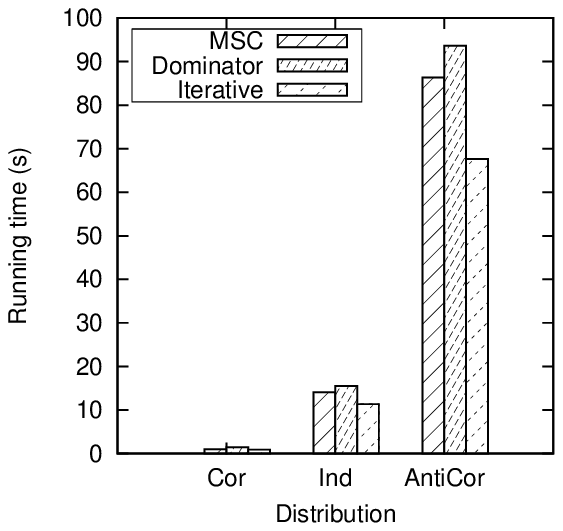} &
	\hspace*{-18mm}
\includegraphics[width=\figwidth]{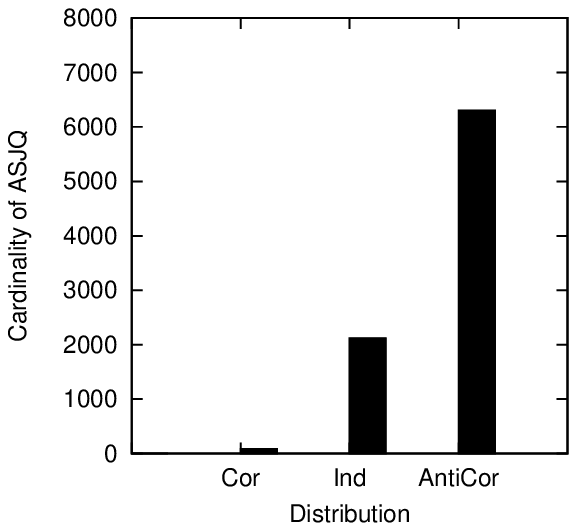} \\
	\hspace*{-10mm}
(a) Running time. &
	\hspace*{-18mm}
(b) Cardinality of ASJQ.
\end{tabular}
\end{center}
\caption{Effect of dataset distribution.}
\label{fig:distr}
\end{figure}

We measured the effect of three standard data distributions---correlated,
independent, and anti-correlated---on the ASJQ algorithms.  The results are
shown in Figure~\ref{fig:distr}.  The cardinality for the correlated dataset is
very small, while that for the anti-correlated dataset is quite large.  In a
perfectly correlated dataset, there is only one skyline record, which dominates
all other records.  In a perfectly anti-correlated dataset, every record is in
the skyline set.  The independent dataset is mid-way, and the cardinality
depends on the dimensionality.  This behavior is reflected in the results.

For the correlated and the independent datasets, the running times of the three
algorithms are similar, while for the anti-correlated dataset, the iterative
algorithm shows an advantage, as it processes the large dominator sets
progressively by only comparing it with certain target sets.

\subsection{Effect of number of categories of join attribute}

\begin{figure}[t]
\begin{center}
\begin{tabular}{cc}
	\hspace*{-10mm}
\includegraphics[width=\figwidth]{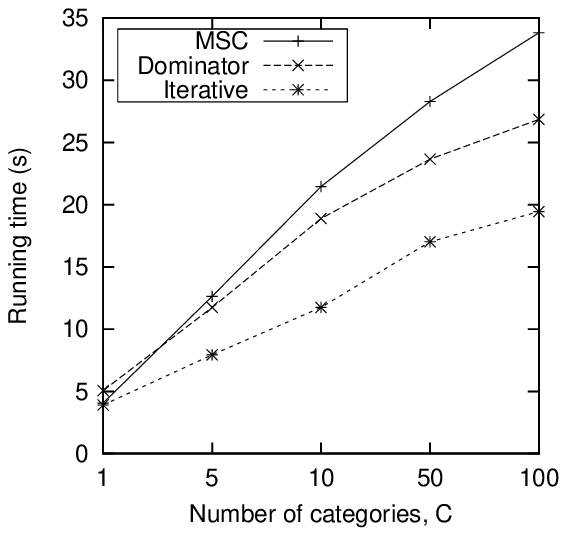} &
	\hspace*{-18mm}
\includegraphics[width=\figwidth]{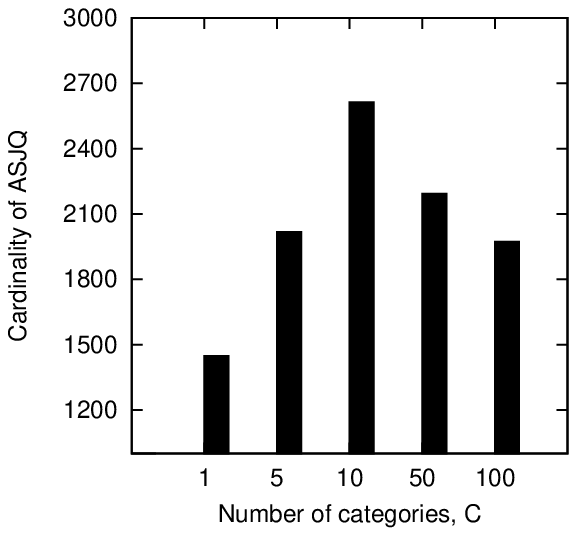} \\
	\hspace*{-10mm}
(a) Running time. &
	\hspace*{-18mm}
(b) Cardinality of ASJQ.
\end{tabular}
\end{center}
\caption{Effect of number of categories of join attribute.}
\label{fig:cat}
\end{figure}

The final experiment on synthetic data measures the effect of the number of
categories of the join attribute.  We assume that only one attribute used for
joining the two relations, and the join condition is an equi-join.  The number
of categories signifies the possible values of the join attribute.

For datasets with cardinality $N$ and number of categories $C$, assuming an
uniform distribution of the join attribute, the total cardinality of the joined
relation is $C \times (N / C)^2 = N^2 / C$.  Hence, as $C$ increases, the
cardinality decreases.  When $C = 1$, the join degenerates to a Cartesian
product of the two relations with $N^2$ tuples.

The cardinality of the ASJQ, however, does not decrease with $C$.  As shown in
Figure~\ref{fig:cat}(b), it attains a maximum in the middle.  When $C$ is low,
even though the number of tuples is high, the chance of a tuple dominating
others is higher as the join attribute is same for more number of tuples.  At
higher values of $C$, the number of joined tuples becomes small, leading to
lower ASJQ cardinality.

Figure~\ref{fig:cat}(a) shows that regardless of the cardinality, the running
time increases with increasing $C$.  When $C$ is more, the initial full skyline
sets ($A_0$ and $B_0$) are larger as there is less probability of a tuple
matching another tuple in the join attribute, and therefore, dominating it.
Consequently, the latter stages of the algorithm are affected and the running
time increases.

\subsection{Real Datasets}
\label{subsec:real}

In this section, we evaluate the performance of the ASJQ algorithms for a real
dataset.  The real dataset consists of the statistics of basketball players
obtained from \url{http://www.databasebasketball.com/}.  The cardinality of the
dataset was $N=10^4$ with $3$ local attributes ($L=3$) and $2$ aggregate
attributes ($G=2$).  We performed a self-join of the dataset with the join
condition as equality.

Four settings were created by varying the number of join attributes.  In setting
1, \emph{year} was used as the join attribute, while in setting 2, the dataset
was joined on the \emph{team}.  For setting 3, no join attribute was used, which
corresponds to the Cartesian product of the relations.  Setting 4 used both the
attributes for joining.

The results are summarized in Figure~\ref{fig:real}.  The cardinality of the
final ASJQ result set was the highest when no join attribute was used (setting
3) and was lowest when both the attributes were used (setting 4).  The running
times reflected the trends of the cardinalities.  The iterative algorithm
performed the best, followed by the dominator-based approach.  The MSC algorithm
was the slowest.  The strategy of the iterative algorithm to prune progressively
proved to be the best.

\begin{figure}[t]
\begin{center}
\begin{tabular}{cc}
	\hspace*{-10mm}
\includegraphics[width=\figwidth]{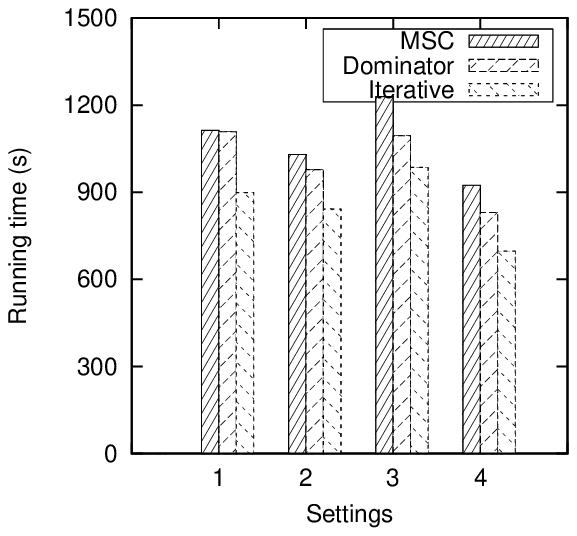} &
	\hspace*{-18mm}
\includegraphics[width=\figwidth]{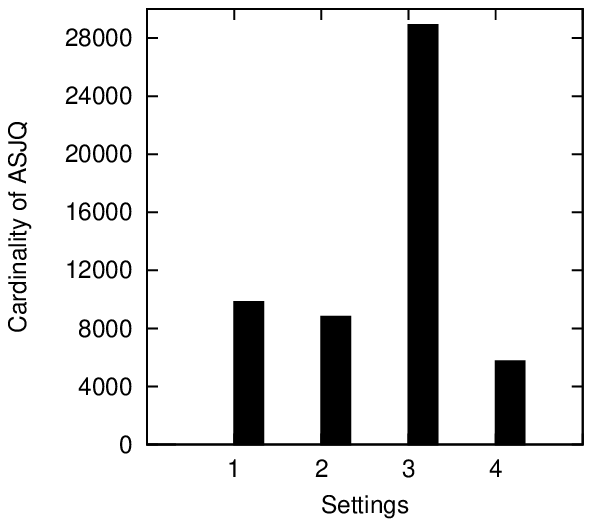} \\
	\hspace*{-10mm}
(a) Running time. &
	\hspace*{-18mm}
(b) Cardinality of ASJQ.
\end{tabular}
\end{center}
\caption{Real datasets.}
\label{fig:real}
\end{figure}

\section{Conclusions}
\label{sec:concl}

In this paper, we have proposed a novel query, the \textsc{Aggregate Skyline
Join Query (ASJQ)}.  This extends the general skyline operator to multiple
relations involving joins using aggregate operations over attributes from
different relations.  The ASJQ processing is explained with the MSC approach,
dominator-based approach and the iterative approach, in addition to the na\"ive
algorithm.  Extensive experiments confirm that our algorithms perform well with
real datasets, and also scale nicely with dimensionality and cardinality of the
relations.  In future, we would like to extend ASJQ to distributed environments
and devise parallel algorithms to process the queries more efficiently.

{
\bibliographystyle{abbrv}
\balance
\bibliography{refer}
}

\end{document}